\documentclass{llncs}

\usepackage{amssymb} 

\usepackage[vlined,linesnumbered,ruled]{algorithm2e}
\usepackage{mathtools}
\usepackage{tikz}
\usepackage{hyperref}

\DeclareMathOperator{\dist}{dist}
\DeclareMathOperator{\ecc}{ecc}
\DeclareMathOperator{\conv}{conv}

\begin{document}
	
\title{Characterising AT-free Graphs with BFS}
\author{Jesse Beisegel}

\institute{Brandenburg University of Technology}


\maketitle

\begin{abstract}	
	An asteroidal triple free graph is a graph such that for every independent triple of vertices no path between any two avoids the third. In a recent result from Corneil and Stacho, these graphs were characterised through a linear vertex ordering called an AT-free order. Here, we use techniques from abstract convex geometry to improve on this result by giving a vertex order characterisation with stronger structural properties and thus resolve an open question by Corneil and Stacho. These orderings are generated by a modification of BFS which runs in polynomial time. Furthermore, we give a linear time algorithm which employs multiple applications of (L)BFS to compute AT-free orders in claw-free AT-free graphs and a generalisation of these.

\end{abstract}

\section{Introduction}
In a classical paper of algorithmic graph theory by Lekkerkerker and Boland from the early 1960s~\cite{lekkekerker1962representation} the authors used a forbidden substructure called an \emph{asteroidal triple} to characterise interval graphs. An asteroidal triple is an independent triple of vertices, such that for any two of them there is a path that avoids the third. This definition gave rise to the introduction of the class of asteroidal triple free graphs (AT-free graphs) and due to the fact that these graphs form a superclass of both the interval and cocomparability graphs, there has been considerable research interest for the last two decades.

AT-free graphs are widely believed to exhibit a "linear structure"~\cite{kohler2015linear} akin to the interval graphs and two results in particular corroborate this claim: In~\cite{corneil1997asteroidal} it was shown that every AT-free graph contains a \emph{dominating pair}, i.e., a pair of vertices such that every path between them forms a dominating set for the whole graph. This result was strengthened in the same paper~\cite{corneil1997asteroidal} which characterised AT-free graphs with the so-called \emph{spine property}: A graph $ H $ has the spine property, if for every non-adjacent dominating pair $ s $ and $ t $ there exists a neighbour of $ t $, say $ t' $, such that $ s $ and $ t'$ are a dominating pair in the connected component of $ H-t $ that contains $ s $. As shown in~\cite{corneil1997asteroidal}, a graph $ G $ is an asteroidal triple free graph if and only if every connected induced subgraph of $ G $ has the spine property. This can be seen as a generalisation of the fact that the maximal cliques of interval graphs form a chain.

An important algorithmic tool in the theory of interval graphs has been their characterising linear vertex ordering, the \emph{interval order}. This is a linear ordering $ \tau = (v_1, \ldots ,v_n) $ of the vertices of a graph $ G=(V,E) $ such that for $ u \prec_{\tau} v \prec_{\tau} w $ and $ uw \in E $ we have $ uv \in E $. It was long conjectured that such a characterising linear vertex ordering must also exist for AT-free graphs and while in a recent result~\cite{corneil2015vertex} this conjecture was answered in the positive, the notion of these orderings leaves quite a bit of freedom.

Ideally, such an ordering would somehow capture the structure given in the spine property in~\cite{corneil1997asteroidal} (as it is in the case of interval orderings which immediately gives us the chain of maximal cliques). However, the so-called \emph{LexComp} ordering that is constructed in~\cite{corneil2015vertex} has one significant drawback: For some graphs the resulting ordering is "folded" in a way that seems to contradict our notion of linear behaviour. For example, given the path graph with $ 2n+1 $ vertices, the $ P_{2n+1} $, where the vertices are numbered from left to right along the path, we would expect any viable linear vertex ordering to be $ (1,2,\ldots ,2n+1) $ or its inversion. The algorithm in~\cite{corneil2015vertex}, on the other hand, might output $ (n+1,n,n+2,n-1, \ldots , 1, 2n+1) $. In addition, this construction can even yield vertex orders $ \tau := (v_1, \ldots ,v_n) $ such that there are $ i \in \{1, \ldots ,n \} $ for which $ G[v_1, \ldots ,v_i] $ is not connected - for example the circuit in five vertices, i.e., $ C_5 $. More examples can be found in Figure \ref{fig1}.

In an attempt to remedy this issue, the authors of~\cite{corneil2015vertex} investigate whether it is possible to find AT-free orderings that coincide with search orders. After proving that there are graphs $ G $ such that no LBFS ordering of $ G $ is an AT-free order, they conjecture that every AT-free graph has an AT-free order that is a BFS order.

\begin{conjecture}\cite{corneil2015vertex}
	Let $ G=(V,E) $ be an AT-free graph. Then there exists a BFS ordering $ \tau = (v_1, \ldots ,v_n) $ that is an AT-free order. 
\end{conjecture}

We will prove an even stronger version of this conjecture, and show how such an order can be used to wed the notion of an AT-free ordering to the spine property. We will also give a polynomial time algorithm to compute such an order that takes approximately the same time as the previous best known algorithm to compute AT-free orders, i.e. $ \mathcal{O}(nm) $~\cite{corneil2015vertex}. The best known algorithm to recognise AT-free graphs uses fast matrix multiplication and takes $ \mathcal{O}(n^{2.82}) $ time~\cite{kratsch2006between} and it can be shown that recognition of AT-free graphs is at least as hard as recognising graphs without an independent set of size three~\cite{spinrad2003efficient}.

For the special case of claw-free AT-free graphs and a generalisation of these we give linear time algorithms to compute AT-free (L)BFS orders. This is a surprising result, as it was shown in~\cite{hempel2002claw} that the recognition of claw-free AT-free graphs is at least as hard as triangle recognition. This dichotomy is of striking resemblance to the case of comparability graphs, where a characterising linear ordering in the form of a transitive orientation can be found in linear time, while there is no known recognition algorithm that is faster than matrix multiplication~\cite{spinrad2003efficient}. Due to these facts, we conjecture that it is possible to compute AT-free orderings in linear time in the general case using some form of modified breadth-first-search. As is the case for comparability graphs, such a linear ordering might then be used for linear time optimisation algorithms that are robust for AT-free graphs, i.e. which can be applied without solving recognition first (for further information on robust algorithms see~\cite{spinrad2003efficient}).

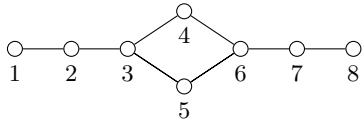
\begin{figure}
\begin{minipage}{0.4\textwidth}	
	\begin{tikzpicture}[vertex/.style={inner sep=2pt,draw,circle},auto]
	
	\node[vertex,label=below:$1$] (1) at (0,0){};
	\node[vertex,label=below:$2$] (2) at (0.75,0){};
	\node[vertex,label=below:$3$] (3) at (1.5,0){};
	\node[vertex,label=below:$4$] (4) at (2.25,0.5){};
	\node[vertex,label=below:$5$] (5) at (2.25,-0.5){};
	\node[vertex,label=below:$6$] (6) at (3,0){};
	\node[vertex,label=below:$7$] (7) at (3.75,0){};
	\node[vertex,label=below:$8$] (8) at (4.5,0){};
	
	\draw[] (1)--(2)--(3)--(4)--(6)--(5)--(3)--(5)--(6)--(7)--(8);
	
	\end{tikzpicture}	
\end{minipage}
\hfill
\begin{minipage}{0.5\textwidth}
	\begin{tabular}{ll}
		Arbitrary AT-free: & $ (4,5,2,7,3,6,1,8) $ \\
		LexComp: & $ (4,5,3,6,2,7,1,8) $ \\
		$ \text{BFS}^{\conv}(G,1)$: & $(1,2,3,4,5,6,7,8)$ \\
	\end{tabular}
\end{minipage}
\caption{Graph with its various AT-free orders}
\label{fig1}
\end{figure}

\section{Preliminaries}

In the following, we will exclusively refer to simple connected graphs $ G $ with vertex set $ V $ and edge set $ E $. The \emph{neighbourhood} of $ v $ in $ G $ is the set $ N_G(v) := \{w ~:~ vw \in E\} $ and $ N[v] := N(v)+v$. A vertex with only one neighbour in $ G $ will be called a \emph{pendant vertex}. A \emph{walk} $ W  $ of length $ k $ in $ G $ is a succession of vertices $ (v_1, \ldots , v_{k+1}) $ such that $ v_iv_{i+1} \in E $ for all $ i \in \{1, \ldots , k\} $. If a walk $ P $ has the additional property that all vertices are distinct, we call $ P $ a \emph{path}. We say that a path $ P $ \emph{avoids} a vertex $ v $, if $ v $ does not have any neighbours on $ P $, while a vertex $ v $ \emph{intercepts} a path $ P $ if it has at least one neighbour on $ P $. 

The \emph{distance} between two vertices $ s $ and $ t $ is the length of a shortest path between these vertices and will be denoted by $ \dist_G(s,t) $. The set of vertices that have distance $ k $ to a vertex $ s $ is called the \emph{$ k $-th distance layer} from $ s $ of $ G $ and is denoted by $ L_G^k(s) $. For every vertex $ v \in V $ we say that $ N_s^k(v):= L_G^k(s) \cap N_G(v) $. A vertex $ x $ with largest distance from $ s $ is called \emph{eccentric} with respect to $ s $ and its distance to $ s $ is the \emph{eccentricity} $ \ecc_G(s) $ of $ s $. The \emph{eccentricity} of $ G $ is the largest such value among all vertices.

A subset $ D \subseteq V $ is called a \emph{dominating set} of $ G $ if every vertex in $ V $ has a neighbour in $ D $. If the set $ D $ forms a path in $ G $ it is called a \emph{dominating path}. Two vertices $ s $ and $ t $ of $ G $ form a \emph{dominating pair}, if every path between them is dominating. A permutation $ \tau := (v_1, \ldots , v_n) $ of the vertices of $ G $ will be called a \emph{linear vertex ordering}.

Given a linear vertex ordering $ \tau $ we can formulate a derivative of \emph{Breadth First Search} called BFS+($ \tau $). This algorithm is a breadth first search which prioritises vertices that are further to the right in $ \tau $, i.e. at any point of the search where neighbours of the current vertex are added to the queue, the vertices with highest $ \tau $-value are added first.

Lexicographic Breadth First Search (Algorithm~\ref{lbfs}) was introduced in~\cite{rose1976algorithmic} to recognise chordal graphs and has been an important ingredient in many recognition and optimisation algorithms since.

\begin{algorithm}
	\KwIn{Connected graph $G=(V,E)$ and a distinguished vertex $ s \in V $}
	\KwOut{A vertex ordering $ \tau $}
	\Begin{
		$ label(s) \leftarrow n $\;
		
		\For{each vertex $v \in V-{s}$}{$ label(v) \leftarrow \emptyset $\;}
		
		\For{$ i \leftarrow 1 $ to $ n $}{pick an unnumbered vertex $ v $ with lexicographically largest label\;\label{slice}
			$ \tau(i) \leftarrow v$\;
			\For{each unnumbered vertex $ u \in N(v) $}{append $ (n-i) $ to $ label(w) $\;}}
		
	}\caption{LBFS}
	\label{lbfs}
\end{algorithm}

If two vertices have the same label in step~\ref{slice}, we say that they are \emph{tied}. We call a set of tied vertices $ S $ encountered in step~\ref{slice} of Algorithm~\ref{lbfs} a \emph{slice}. Given an LBFS order $ \tau $ and two vertices $ u $ and $ v $ with $ u \prec_{\tau} v $, we denote the vertex-minimal slice with respect to $ \tau $ containing $ u $ and $ v $ as $ \Gamma_{u,v}^{\tau} $.

As before with the BFS, given a linear vertex order $ \tau $, we can define an LBFS+($ \tau $) in the following way: At any point in the search at which we encounter a slice, i.e. a set of tied vertices, the vertex of highest $ \tau $-value is chosen first.

There are many interesting properties and applications of LBFS, and some of these can be found in~\cite{corneil2004lexicographic}. Here we will need one result in particular, which is a useful tool for the analysis of LBFS and LBFS+ orders.

\begin{lemma}[Prior Path Lemma]\cite{corneil2009lbfs}
	Let $ \tau $ be an arbitrary LBFS of a graph $ G $ and let $ u,v \in V $ with $ u \prec_{\tau} v $. Let $ w $ be the $ \tau $-first vertex of the connected component $ C_u $ of $ \Gamma^{\tau}_{u,v} $ containing $ u $. There exists a $ w $-$ u $-path in $ \Gamma^{\tau}_{u,v} $ all of whose vertices, with the possible exception of $ u $, are not adjacent to $ v $. Moreover, all vertices on this path, other than $ u $, occur before $ u $ in $ \tau $. Such a path is called a \emph{prior path}.
\end{lemma}

Finally, a graph will be called \emph{claw-free}, if it does not contain a claw graph, i.e. the $ K_{1,3} $, as an induced subgraph. We will call the three independent vertices the \emph{prongs} and the fourth vertex the \emph{base} of the claw.

\section{Convex Geometries and AT-free Graphs}

\begin{definition}\cite{edelman1985theory}
A set $ V $ and a family of subsets $ \mathcal{C} $ of $ V $ form a \emph{convexity space}, if $ \emptyset , V \in  \mathcal{C} $ and $ \mathcal{C} $ is closed under intersection. The smallest convex set $ \conv(X) $ containing a set $ X \subseteq V $ is called the \emph{convex hull} of $ X $.
We say that a convexity space $ (V, \mathcal{C}) $ is a \emph{convex geometry}, if for every convex set $ X $ and two points $ p,q \in V \backslash X$:
\begin{equation*}
q \in \conv (X+p) \Rightarrow p \notin \conv (X+q).
\end{equation*}
This is sometimes referred to as the \emph{anti-exchange property}. A convex set whose complement is also convex is called a \emph{halfspace}.
\end{definition}

The anti-exchange property motivates an ordering of the ground set $ V $ of a convex geometry: An ordering $ \tau =(v_1, \ldots , v_n) $ is a \emph{convexity ordering}, if $ \{v_1, \ldots v_i\} $ is convex for every $ i \in \{1, \ldots , n\} $. If $ \{v_1, \ldots ,v_i \} $ is a halfspace for every $ i \in \{1, \ldots , n\} $, then we call $ \tau $ a halfspace ordering.

One way to define a convexity space is through strict betweenness. Following~\cite{chvatal2009antimatroids} we say that a \emph{strict betweenness} over a ground set $ V $ is a ternary relation $ \mathcal{B} \subset V^3 $ such that
\begin{equation*}
(a,b,c) \in \mathcal{B} \text{ implies that } (c,b,a) \in \mathcal{B} \text{ and } a,b, \text{ and } c \text{ are distinct.}
\end{equation*}

The convexity space with regard to this betweenness is then defined to be the pair $ (V,\mathcal{C}_{\mathcal{B}}) $ where
\begin{equation*}
\mathcal{C_{\mathcal{B}}}:= \{C \subseteq V ~:~ \{a,c\} \subseteq C \text{ and } (a,b,c) \in \mathcal{B} \text{ implies }  b \in C  \}.
\end{equation*}
On graphs we can define just such a strict betweenness on the set of vertices and thus we can construct a convexity space in the following way:

\begin{definition}
	Given a graph $ G=(V,E) $ we say that $ (x,y,z) \in \mathcal{B}_{D}(G) $, if there is a chordless $ x$-$y $-path that avoids $ z $ and a chordless $ y$-$z $-path that avoids $ x $. The set of vertices $ y $ with $ (x,y,z) \in \mathcal{B}_{D}(G) $ is called the domination interval of $ x $ and $ z $ and is denoted by $ I_{D}(x,z) $. The ternary relation $ \mathcal{B}_D(G) $ is called the \emph{domination betweenness} of $ G $ and it is easy to see that this is a strict betweenness. As a result, we obtain a convexity space $ (V,\mathcal{C}_{\mathcal{B}_D}(G)) $ which we will call the \emph{domination convexity} of $ G $.
\end{definition}

A vertex $ y $ is said to be \emph{admissible}, if there are no two vertices $ x $ and $ z $ such that $ (x,y,z) \in \mathcal{B}_{D}(G)$. An \emph{AT-free ordering} is an ordering $ \tau = (v_1, \ldots , v_n) $ of the vertices such that for any $ (x,y,z) \in \mathcal{B}_{D}(G) $ we have $ y \prec_{\tau} x $ or $ y \prec_{\tau} z $. It is easy to see that for any such ordering $ \{v_1, \ldots , v_i\} $ is domination convex for any $ i \in \{1, \ldots , n \} $. If $ \tau $ is such that for any $ (x,y,z) \in \mathcal{B}_{D}(G) $ we have $x \prec_{\tau} y \prec_{\tau} z $ we say that it is a \emph{bilateral AT-free ordering} of $ G $.

The connection between convexity theory and AT-free graphs was recently made in~\cite{chang2015gray} and~\cite{chang2017convex} and it was furthermore shown that the convexity space thus defined is in fact a convex geometry. In the following we have bundled that result with a number of other characterising properties of AT-free graphs:

\renewcommand{\labelenumi}{(\roman{enumi})}
\renewcommand{\theenumi}{(\roman{enumi})}

\begin{theorem}\cite{broersma1999independent}\cite{chang2015gray}\cite{chang2017convex}\cite{corneil1997asteroidal}\cite{corneil2015vertex}\cite{edelman1985theory}\cite{koehler1999graphs}
	Given a graph $ G $, its domination betweenness $ \mathcal{B}_D $ and its domination convexity $ \mathcal{C}_{\mathcal{B}_D} $, the following statements are equivalent:
	\begin{enumerate}
		\item $ G $ is AT-free.
		\item If $ (w,x,y) \in \mathcal{B}_D(G) $ and  $ (x,y,z) \in \mathcal{B}_D(G) $ then  $ (w,x,z) \in \mathcal{B}_D(G) $, i.e., $ \mathcal{B}_D(G) $ is a transitive ternary relation.
		\item Every connected induced subgraph of $ G $ has the spine property.
		\item $ G $ has an AT-free order.
		\item $ (V,\mathcal{C_{\mathcal{B}_D}}(G)) $ is a convex geometry.			
	\end{enumerate}
\end{theorem}

\section{AT-free BFS-Orders}

\begin{theorem}\label{theorem2}
Let $ G $ be a connected AT-free graph. Then for any vertex $ s \in V $ there is a linear vertex order $ \tau := (s=v_1, \ldots , v_n) $ that is an AT-free order and a BFS order.
\end{theorem}
\begin{proof}
Let $ \tau $ be a BFS order starting in an arbitrary vertex $ s $ of $ G $ with the following tie-break rule: At each step $ i $ choose the vertex $ v_i $ such that $ \conv (\{s=v_1, \ldots, v_i\}) $ has smallest cardinality among all allowed choices at step $ i $. We will show, that $ \{s=v_1, \ldots , v_i\} $ is convex for $ i \in \{1, \ldots , n\} $, which implies that $ \tau $ is an AT-free order. The proof will be by induction on the BFS steps.

For $ k =1$ the claim is true, as every one element set is convex in $ \mathcal{C} $.

We show the claim for step $ k $, assuming it is true for $ k-1 $. Suppose $ v_k $ is chosen. Then $ \{v_1, \ldots, v_{k-1}\} $ is convex and $ v_k $ is such that $ \conv(\{v_1,\ldots, v_{k-1}\} + v_k) $ is smallest among all vertices that can be chosen by the search in step $ k $. As we are conducting a BFS there is a vertex $ y \in \{v_1, \ldots v_{k-1}\} $ that is adjacent to all possible choices, but no others. Assume that $ \{v_1, \ldots ,v_k\} $ is not convex. Then there is a vertex $ p \in V \backslash \{v_1, \ldots, v_k\} $, such that $ (v,p,v_k) \in \mathcal{B}_D $ for some vertex $ v \in \{v_1, \ldots v_{k-1}\} $. As $ (V,\mathcal{C_{\mathcal{B}_D}}(G)) $ is a convex geometry, we can deduce that $ \conv(\{v_1, \ldots, v_{k-1}\}+p) \subsetneq \conv(\{v_1, \ldots v_{k-1}\}+v_k)$. This implies that $ yp \notin E $ due to the choice of $ v_k $. Let $ w $ be the vertex that forced $ v $ into the BFS ordering (it may be that $ y=w $). Due to the definition of BFS we see that $\dist_G(s,w) \leq \dist_G(s,y) <  \dist_G(s,p) $. We can assume that $ wp \notin E $, as otherwise $ p $ would have been chosen before $ v_k $. Therefore, the vertices $ \{v,v_k,p\} $ form an asteroidal triple, due to the $ p $-avoiding walk from $ v $ to $ v_k $ along $ w $, $ s $ and $ y $. This is a contradiction to fact that $ G $ is AT-free.

\qed
\end{proof}

This theorem implies an algorithm for computing an AT-free BFS order which will be denoted by $\text{BFS}^{\conv}$.

\begin{algorithm}
	\KwIn{Connected graph $G$ and a distinguished vertex $ s \in V $}
	\KwOut{A vertex ordering $ \sigma $}
	\Begin{
		Compute $I(v,w)$ for every pair of vertices $v,w \in V $\;
		$ L \leftarrow \{s\} $\;
		$ S \leftarrow \emptyset $\;
		\For{$ i \leftarrow 1 $ to $ n $}{Choose the first vertex $ v $ from $ L $ such that there are no $ u \in S $ and $ z \in V-S $ with $ z \in I(u,v) $\;
			Delete $ v $ from $ L $\;
			$ \sigma(i) \leftarrow v $\;
			$ S \leftarrow S \cup \{v\} $\;
			\For{each unnumbered vertex $ w $ adjacent to $ v $ }{\If{$ w \notin L$}{Append $ w $ to end of $ L $\;}}
		}
	}
	\caption{$ \text{BFS}^{\conv} $}
	\label{bfs-conv}
\end{algorithm}

Any such ordering $ \tau := (v_1, \ldots , v_n) $ obviously has the property that for every $ i \in \{1, \ldots ,n\}  $ the induced subgraph $ G[ \{v_1, \ldots , v_i\}] $ is connected. This is already an improvement on the orders produced by the algorithm given in~\cite{corneil2015vertex} and in Figure~\ref{fig1} we compare orders computed by the different algorithms. On the other hand, returning to the example given in the introduction, the $ P_{2k+1} $ path graph, we can see that starting the $\text{BFS}^{\conv}$ in vertex $ k+1 $ still yields an undesirable order.

Starting in an admissible vertex, which in the case of $ P_{2k+1} $ will be one of the endpoints or one of their neighbours, is an easy remedy of this problem. However, with a little modification to our search routine we can not only solve this issue, but make an intriguing link with the AT-free graphs characterisation through the spine property. We shall call a vertex ordering $ \tau = (v_1, \ldots , v_n) $ a \emph{monotone dominating pair order}, if for every $ i \in \{1, \ldots , n\} $ the vertices $ v_1 $ and $ v_i $ form a dominating pair in the induced subgraph $ G[v_1, \ldots , v_i] $.

\begin{theorem}\cite{corneil1999linear}\label{theorem3}
	Let $ G=(V,E) $ be a connected AT-free graph and suppose that $ s $ is an admissible vertex. Let $ \tau= (v_1, \ldots, v_n) $ be a vertex order produced by LBFS $ (G,s) $. Then for any $ i \in \{1, \ldots ,n\} $ the vertices $ v_1 $ and $ v_i $ form a dominating pair of $ G[v_1, \ldots, v_n] $, i.e., $ \tau $ is a monotone dominating pair order.
\end{theorem}

In the following we will prove an analogous result for $ \text{BFS}^{\conv} $.

\begin{lemma}\label{lemma1}
	Let $ G=(V,E) $ be an AT-free graph and let $ s $ be an admissible vertex of eccentricity $k >2 $. If $ \tau :=(s=v_1, \ldots , v_n=t) $ is the output of $ \text{BFS}^{\conv}(G,s) $, then $ s $ and $ t $ form a dominating pair.
\end{lemma}
\begin{proof}
	Suppose $ s $ and $ t $ are not a dominating pair. Then there is an $ s $-$ t $-path $ P $ and a vertex $ w \in V $ such that $ P $ avoids $ w $. W.l.o.g. we can assume that $ P $ is induced. As $ s $ is admissible and $ sw, st \notin E $ we must assume that $ t $ sees every $ w $-$s$-path. Therefore $ w $ must be in the distance layer $  L_G^k(s) $ and $ N_s^{k-1}(w) \subseteq N_s^{k-1}(t) $. As $ k > 2 $, we can deduce that $ (w,t,s) \in \mathcal{B}_D(G) $ which is a contradiction to $ \tau $ being an AT-free order.
	
	\qed
\end{proof}

However, applying a $ \text{BFS}^{\conv}$ with an admissible start vertex must not always result in a monotone dominating pair order, as can be seen in Figure~\ref{fig2}. 

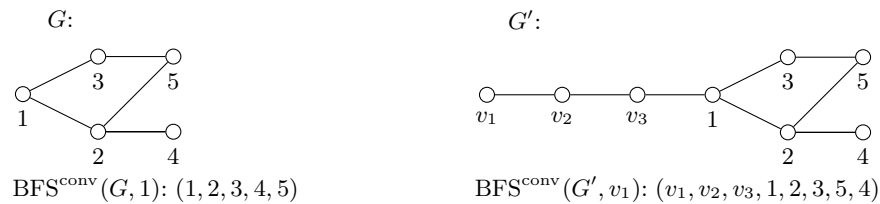
\begin{figure}
	\begin{minipage}{0.3\textwidth}
	\begin{tikzpicture}[vertex/.style={inner sep=2pt,draw,circle},auto]
	
	\node[] (G) at (0.5,1.5) {$G$:};
	
	\node[vertex,label=below:$1$] (1) at (0,0.5){};
	\node[vertex,label=below:$2$] (2) at (1,0){};
	\node[vertex,label=below:$3$] (3) at (1,1){};
	\node[vertex,label=below:$4$] (4) at (2,0){};
	\node[vertex,label=below:$5$] (5) at (2,1){};

	\draw[] (1)--(2)--(4)--(2)--(5)--(3)--(1);
	
	\end{tikzpicture}

	\end{minipage}
\hfill
	\begin{minipage}{0.5\textwidth}
		\begin{tikzpicture}[vertex/.style={inner sep=2pt,draw,circle},auto]
		
		\node[] (G) at (-2.5,1.5) {$G'$:};
		
		\node[vertex,label=below:$v_1$] (v1) at (-3,0.5){};
		\node[vertex,label=below:$v_2$] (v2) at (-2,0.5){};
		\node[vertex,label=below:$v_3$] (v3) at (-1,0.5){};
		\node[vertex,label=below:$1$] (1) at (0,0.5){};
		\node[vertex,label=below:$2$] (2) at (1,0){};
		\node[vertex,label=below:$3$] (3) at (1,1){};
		\node[vertex,label=below:$4$] (4) at (2,0){};
		\node[vertex,label=below:$5$] (5) at (2,1){};
		
		\draw[] (v1)--(v2)--(v3)--(1)--(2)--(4)--(2)--(5)--(3)--(1);
		
		\end{tikzpicture}
	\end{minipage}

\begin{minipage}{0.4\textwidth}
	\begin{tabular}{ll}
		$ \text{BFS}^{\conv}(G,1)$: & $( 1, 2, 3, 4, 5)$ 
	\end{tabular}
\end{minipage}
\hfill
\begin{minipage}{0.5\textwidth}
	\begin{tabular}{ll}
			\begin{tabular}{ll}
				$ \text{BFS}^{\conv}(G',v_1)$: & $( v_1, v_2, v_3, 1, 2, 3, 5, 4)$ 
			\end{tabular}
	\end{tabular}
\end{minipage}
	\caption{Graph for which $ \text{BFS}^{\conv}$ does not necessarily output a monotone dominating pair ordering and the graph $ G' $ constructed from $ G $ as in Theorem~\ref{theorem4}.}
	\label{fig2}
\end{figure}

In~\cite{corneil1997asteroidal} it is shown that for an AT-free graph $ G $ and an admissible vertex $ s $ the graph $ G' $ obtained by adding a pendant vertex $ v $ to $ s $ is also AT-free and $ v $ is admissible in $ G' $. With this operation we can artificially raise the eccentricity of our starting vertex and generalise Lemma~\ref{lemma1} to all AT-free graphs.

\begin{theorem}\label{theorem4}
	Let $ G $ be a connected AT-free graph. For every admissible vertex $ s $ there is a vertex ordering $ \tau $ beginning in $ s $ that is both AT-free and a monotone dominating pair ordering.
\end{theorem}
\begin{proof}
	We construct an auxiliary graph by adding a three vertex path to $ s $ in the following way: $ G'=(V+\{v_1,v_2,v_3\}, E +\{v_1v_2,v_2v_3,v_3s\}) $. As $ s $ is admissible, the graph $ G' $ is again AT-free and $ v_1 $ is admissible in $ G' $ with $ \ecc_{G'}(v_1) > 2 $. The order $ \tau' = (v_1,v_2,v_3,w_1, \ldots , w_n)$ that is generated by $ \text{BFS}^{\conv}(G',v_1) $ is an AT-free order and with Lemma~\ref{lemma1} it is easy to see that $ \tau = (w_1, \ldots ,w_n) $ is a monotone dominating pair order for $ G $.
	
	\qed
\end{proof}

\section{AT-free Orders in Claw-free AT-free graphs}

After having established the existence of AT-free BFS orders and a polynomial-time algorithm for their computation, we are interested in finding a simple linear time algorithm. In many graph classes, forbidding induced claw-graphs yields strong structural properties for BFS searches. For example, in~\cite{corneil2004simple} and in~\cite{meister2005recognition} the authors use these structural properties to generate unit interval respectively minimal triangulation orderings. As in the papers cited above, we will use successive applications of BFS as well as LBFS.

\begin{lemma}\label{lemma4}
	Let $ G $ be claw-free and AT-free. Then the last vertex of a BFS is admissible.
\end{lemma}
\begin{proof}
	Let $ s $ be the first and $z$ the last vertex of the BFS and let $ k := \dist_G(s,z) $.  Suppose there are $ a,b \in V $ such that $ (a,z,b) \in \mathcal{B}_D(G) $. As $ G $ is AT-free, at least one of $a$ or $b$ must be in the last layer $ L_G^{k}(s) $ of the BFS, w.l.o.g. this is $ a $. If $ \dist_G(s,b) < \dist_G(s,z) $, then $ N_s^{k-1}(a) \subseteq N_s^{k-1}(z)$, as otherwise there is a $ z $-avoiding $ a $-$ b $-path. If $ \dist_G(s,b) = \dist_G(s,a) = \dist_G(s,z)$, then either $ N_s^{k-1}(a) \subseteq N_s^{k-1}(z)$ or $ N_s^{k-1}(b) \subseteq N_s^{k-1}(z)$, as $ G $ is AT-free, and without loss of generality we can assume this to be true for $ a $. Therefore, $ a $ and $ z $ have a common neighbour $ c $ in $ L_G^{k-1}(s) $. If $ c $ is not the start vertex of the BFS, then $ c $ has a neighbour $ d $ in $ L_G^{k-2} $ and $ a,z,c,d $ form a claw.
	If $ c $ is the start vertex, then $ b $ must also be adjacent to $ c $ and $ a,b,c,d $ form a claw.
	
	\qed
\end{proof}

\begin{lemma}\label{lemma5}
	Let $ G $ be a claw-free, AT-free graph and let $ s \in V$ be admissible in $ G $ and $ t $ eccentric with respect to $ s $. Then all but the first distance layers of $ s $, i.e., $ L_G^{0}(s),L_G^{2}(s), \ldots , L_G^{k}(s) $, with $ k=\ecc_G(s) $, are cliques and $ s $ and $ t $ form a dominating pair.
\end{lemma}
\begin{proof}
	For $ L_G^0(s) $ this is obvious. Let $i \geq2$ and suppose there are $ a,b \in L_G^i(s) $ with $ ab \notin E $. As $ s $ is admissible, without loss of generality $ N_s^{i-1}(a) \subseteq N_s^{i-1}(b) $. Therefore $ a $ and $ b $ have a common neighbour $ c \in L_G^{i-1} $. This $ c $ in turn has a neighbour $ d \in L_G^{i-2} $ and $ a,b,c,d $ form a claw, which is a contradiction to the assumption.
	
	As any path $ P $ between $ s $ and $ t $ has one vertex from each distance layer $ L_G^i(s) $ and $ s $ is adjacent to all vertices in $ L_G^1(s) $ they must form a dominating pair.
	
	\qed
\end{proof}

\begin{theorem}\label{theorem1}
	Let $ G $ be an AT-free, claw-free graph. Then a BFS starting in an admissible vertex yields an AT-free order that is a monotone dominating pair order.
\end{theorem}
\begin{proof}
	Let $ \tau $ be such a BFS on $ G $ starting in an admissible vertex  $ s $. Suppose $ (a,z,b) \in \mathcal{B}_D(G) $ and $ a, b \prec_{\tau} z$. We can assume that $ a $, $ b $ and $ z $ do not have the same distance to $ s $ (otherwise we can construct a claw as above). As $ G $ is AT-free, on the other hand, at least one of $ a $ or $ b $ must be in the same layer as $ z $. W.l.o.g. we can assume that $ b $ and $ z $ are in the same layer $ L_G^i(s) $ and $ a $ is in layer $ L_G^j(s) $ with $ j < i $. As $ b $ and $ z $ are independent of each other, they must be in the first layer of the BFS. As $ a $ cannot be the start vertex (it is not adjacent to the other two), this is a contradiction. Lemma~\ref{lemma5} states that $ \tau $ must be a monotone dominating pair order.
	
	\qed
\end{proof}

\begin{lemma}\label{lemma3}
	Let $ G=(V,E) $ be a connected graph with a dominating pair $ s $ and $ t $. Let $ u $ and $ v $ be two vertices with $ uv \notin E $ and $ \dist_G(s,u) < \dist_G(s,v) $. Then $ \dist_G(t,u) \geq \dist_G(t,v) $
\end{lemma}
%

\begin{corollary}\label{corollary1}
	Let $G$ be a claw-free AT-free graph. Then $G$ has a bilateral AT-free ordering and this order can be found in linear time.
\end{corollary}
%
%
%
%
%
%

In the proof of Theorem~\ref{theorem1} we can see that the main obstacles are triples of vertices $ a,b,z \in V(G) $ with $ (a,z,b) \in \mathcal{B}_D(G) $ that form the prongs of a claw. This justifies the following:

\begin{definition}
	Let $G$ be a graph and let $a,b,z,c \in V$ induce a claw with base $c$. We will call such a claw a \emph{bad claw}\index{bad claw}, if $(a,z,b) \in \mathcal{B}_D(G)$.
\end{definition}

It seems reasonable to expect that by forbidding such bad claws we will be able to get similar results to the ones above. On the other hand, there are examples of AT-free bad-claw-free graphs for which the above procedure does not yield either an AT-free order nor a bilateral AT-free ordering (see Figure~\ref{fig4}). In particular, Lemma~\ref{lemma4} does not hold in general for these graphs. Therefore, we will use LBFS which guarantees us an admissible vertex as its end-vertex.

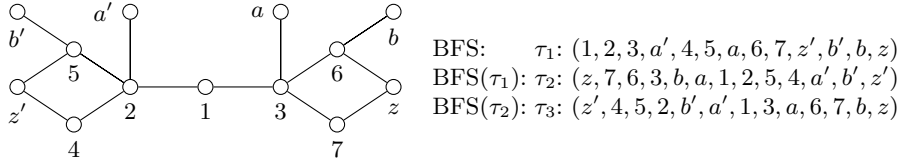
\begin{figure}
	\begin{minipage}{0.45\textwidth}
		\begin{tikzpicture}[vertex/.style={inner sep=2pt,draw,circle},auto]
		
		\node[vertex,label=below:$1$] (1) at (2.5,0){};
		\node[vertex,label=below:$2$] (2) at (1.5,0){};
		\node[vertex,label=below:$3$] (3) at (3.5,0){};
		\node[vertex,label=below:$4$] (4) at (0.75,-0.5){};
		\node[vertex,label=below:$5$] (5) at (0.75,0.5){};
		\node[vertex,label=below:$6$] (6) at (4.25,0.5){};
		\node[vertex,label=below:$7$] (7) at (4.25,-0.5){};
		\node[vertex,label=left:$a$] (a) at (3.5,1){};
		\node[vertex,label=left:$a'$] (a') at (1.5,1){};
		\node[vertex,label=below:$b$] (b) at (5,1){};
		\node[vertex,label=below:$b'$] (b') at (0,1){};
		\node[vertex,label=below:$z$] (z) at (5,0){};
		\node[vertex,label=below:$z'$] (z') at (0,0){};
		
		\draw[] (b')--(5)--(z')--(4)--(2)--(5)--(2)--(a')--(2)--(1)--(3)--(a)--(3)--(6)--(b)--(6)--(z)--(7)--(3);
		
		\end{tikzpicture}
	\end{minipage}	
		\hspace{0pt}
	\begin{minipage}{0.52\textwidth}
		\begin{tabular}{ll}
		BFS: & $ \tau_1 $: $(1,2,3, a', 4, 5, a, 6, 7, z', b', b, z)$ \\
		BFS($ \tau_1 $): & $ \tau_2 $: $ (z, 7, 6, 3, b, a, 1, 2, 5, 4, a', b', z') $ \\
		BFS($\tau_2 $): & $ \tau_3 $: $ (z', 4, 5, 2, b', a', 1, 3, a, 6, 7, b, z) $ 		\end{tabular}
	\end{minipage}		

	\caption{A bad-claw-free graph for which BFS does not yield an AT-free order}
	\label{fig4}
\end{figure}

\begin{lemma}\cite{corneil1999linear}\label{lemma2}
	Let $ G=(V,E) $ be an AT-free graph and let $ \tau $ be an ordering of $ V $ produced by an LBFS. Then the vertex $ t := \tau(n) $ is admissible in $ G $.
\end{lemma}

In fact, the properties of LBFS even make up for the absence of the strong structural property of Lemma~\ref{lemma5} and we can prove analogues to both Theorem~\ref{theorem1} and Corollary~\ref{corollary1}.

\begin{theorem}
	Let $ G $ be AT-free and bad-claw-free. Then an LBFS starting in an admissible vertex yields an AT-free order that is a monotone dominating pair order. 
\end{theorem}
\begin{proof}
	Let $ \tau $ be an LBFS order starting in an admissible vertex $ s $. Suppose $ (a,z,b) \in \mathcal{B}_D(G) $ and $ a,b \prec_{\tau} z $. Without loss of generality, we see that $ i:=\dist_G(s,b)= \dist_G(s,z) $, as $ G $ is AT-free. For that same reason either $ N_s^{i-1}(b) \subseteq N_s^{i-1}(z) $ or $ N_s^{i-1}(a) \subseteq N_s^{i-1}(z) $ or both.
	
	Now suppose $ \dist_G(s,a)=i $. As $ s $ is admissible, and $ a $, $ b $ and $ z $ are independent, they must have a common neighbour $ c $ with $ \dist_G(s,c)=i-1 $ and therefore $ a $, $ b $ and $ z $ and $ c $ form a bad claw, which is a contradiction.
	
	Therefore, we can assume that $ j:= \dist_G(s,a) < i $. With the above we see that $ N_s^{i-1}(b) \subseteq N_s^{i-1}(z)$ and there is a $ b $-avoiding $ a $-$ z $-path $P$. Let $x$ be the $ \tau $-last vertex of $ P $. As $ b \prec_{\tau} z \preceq_{\tau} x $, due to Theorem~\ref{theorem3} the vertex $ b $ must see every $ s $-$x$-path and thus also every $ x$-$a $-path, which is a contradiction. Thus, every LBFS starting in an admissible vertex yields an AT-free order. 
	
	Finally, Theorem~\ref{theorem3} states that every LBFS order of an AT-free graph starting in an admissible vertex is a monotone dominating pair order.
	
	\qed
\end{proof}

\begin{corollary}
	Let $G$ be an AT-free graph that does not have a bad claw as an induced subgraph. Then $G$ has a bilateral AT-free ordering and such an order can be found in linear time.
\end{corollary}

\begin{figure}
	\begin{minipage}{0.3\textwidth}
		\begin{tikzpicture}[vertex/.style={inner sep=2pt,draw,circle},auto]
		
		\node[vertex,label=left:$1$] (1) at (1,3){};
		\node[vertex,label=left:$2$] (2) at (1,2){};
		\node[vertex,label=below:$3$] (3) at (-0.5,1){};
		\node[vertex,label=below:$4$] (4) at (0.5,1){};
		\node[vertex,label=below:$a$] (a) at (0,0){};
		\node[vertex,label=below:$b$] (b) at (2.5,1){};
		\node[vertex,label=below:$c$] (c) at (2,0){};
		\node[vertex,label=below:$z$] (z) at (1.5,1){};
		
		\draw[] (1)--(2)--(3)--(a)--(4)--(2)--(z)--(4)--(z)--(c)--(a)--(c)--(b)--(2);
		
		\end{tikzpicture}
	\end{minipage}
	\hfill
	\begin{minipage}{0.6\textwidth}
		\begin{tabular}{ll}
			LBFS: & $ \tau_1 $: $(1, 2, 4, z, 3, b, a, c)$ \\
			LBFS($ \tau_1 $): & $ \tau_2 $: $(c, a, b, z, 4, 3, 2, 1)$ \\
			LBFS($\tau_2 $): & $ \tau_3 $: $(1, 2, 3, 4, z, b, a, c)$ 
		\end{tabular}
	\end{minipage}
	\caption{Example of a graph with a bad claw. On the right, one can see that the second $ \tau_2 $ is not an AT-free order and $ \tau_3 $ is not a bilateral AT-free order. In fact, this is an example of an AT-free graph that does not possess a bilateral AT-free ordering.}
	\label{fig3}
\end{figure}

These results indicate that a linear time algorithm to construct AT-free orders could also exist for the general case of AT-free graphs. However, none of the techniques used for the (bad-)claw-free graphs can be transferred. In~\cite{corneil2015vertex} it was already shown that there are AT-free graphs which do not possess AT-free orders that are also LBFS orders. In addition, Figure~\ref{fig3} shows a graph which does not possess a bilateral AT-free ordering. Therefore, it will be necessary to use a different search algorithm, possibly a BFS-derivative based on $ \text{BFS}^{\conv} $. We summarise these suppositions in the following:

\begin{conjecture}\label{conjecture2}
	Let $ G =(V,E) $ be an AT-free graph. There is a linear time algorithm that computes an AT-free (BFS) order.
\end{conjecture}

\section{Conclusion}

We resolved an open question from~\cite{corneil2015vertex} by proving that any given AT-free graph has an AT-free order that coincides with a BFS order. The proof implied a polynomial time algorithm for the computation of such an order that is at least as fast as recognition. As a result, we were able to show that there is a close link between the vertex order characterisation of AT-free graphs, and their characterisation through the spine property. As checking whether a vertex order is an AT-free order is in fact of the same difficulty as recognising AT-free graphs, it should still be possible to find AT-free orders in linear time. This could be done by giving a linear time implementation of $ \text{BFS}^{\conv} $ or by constructing another search scheme with similar structural properties.

For the special case of claw-free AT-free graphs we have shown that multiple applications of BFS yield AT-free orders with additional structural properties. In fact, if we exchange generic BFS with LexMinBFS, a derivative defined in~\cite{meister2005recognition}, we can construct an AT-free, monotone dominating pair order that is also a minimal interval completion order. While claw-free AT-free graphs form a strongly restricted subclass of AT-free graphs, it is important to recall that their recognition has been shown to be at least as hard as triangle recognition, the same bound given to the recognition of general AT-free graphs. Furthermore, the results on bad-claw-free graphs can be seen as a first step toward a resolution of Conjecture~\ref{conjecture2}, and give us a strong notion where the algorithmic difficulties lie.

Linear vertex orderings of other graph classes, such as interval orderings or cocomparability orderings, have found many applications in optimisation algorithms on these classes. To the best knowledge of the author, no such results are known with respect to AT-free orderings. By using AT-free BFS orderings such results might be easier to attain. Two of the most likely candidates are the independent set problem and the vertex colouring problem. However, in the case of vertex colouring even for cocomparability graphs there is no known algorithm that utilises the cocomparability ordering. Should it be possible to compute AT-free orders in linear time, it might even be possible to develop robust optimisation algorithms (see~\cite{spinrad2003efficient}) on AT-free graphs, similar to the maximum clique algorithm on comparability graphs.

Finally, it is still an open question whether every AT-free graph admits a DFS order whose reversal is AT-free~\cite{corneil2015vertex}.

\bibliographystyle{plain}

\end{document}